\documentclass[]{sig-alternate}

\usepackage{subfigure}
\usepackage{float}
\usepackage{epsf}\epsfverbosetrue
\usepackage{alltt}

\newtheorem{theorem}{Theorem}

\newif\ifcode
\codefalse
\codetrue

\newcommand{\lit}[1]{\relax\ifmmode
  \mathord{\mathcode`\-="702D\sf #1\mathcode`\-="2200}\else
  {\it #1}\fi}
\newcommand{\ms}[1]{\relax\ifmmode
  \mathord{\mathcode`\-="702D\it #1\mathcode`\-="2200}\else
  {\it #1}\fi}

\ifcode
\usepackage[boxed]{algorithm}
\usepackage{algpseudocode}

\newcommand{\partsize}{}
\newcommand{\size}{\footnotesize}
\newcommand{\smallsize}{\footnotesize}

\algrenewcommand{\algorithmiccomment}[1]{\textsf{\size {\bf //} #1}}
\algnewcommand{\algorithmicto}{\textbf{to}}
\algrenewcommand{\algorithmicindent}{1em}

\algblockdefx[IOASmallPart]{SmallPart}{EndSmallPart}%
	[2]{{\bf \smallsize{#1:}}#2}{\bf Ending}
\algblockdefx[IOAPart]{Part}{EndPart}%
	[2]{{\bf \partsize{#1:}}#2}{\bf Ending}  
\algblockdefx[Precond]{Prec}{EndPrec}%
	{\textbf{Precondition:}}{\bf Ending}
\algblockdefx[Effect]{Eff}{EndEff}%
	{\textbf{Effect:}}{\bf Ending}
\algblockdefx[InputAction]{Input}{EndInput}%
	[2]{\textbf{\size Input #1}%
	    \Eff#2\EndEff}{\bf Ending}
\algblockdefx[OutputAction]{Output}{EndOutput}%
	[3]{\textbf{\size Output #1}%
	    \Prec#2\EndPrec%
	    \Eff#3\EndEff}{\bf Ending}
\algblockdefx[InternalAction]{Internal}{EndInternal}%
	[3]{\textbf{\size Internal #1}%
	    \Prec#2\EndPrec%
	    \Eff#3\EndEff}{\bf Ending}
\algblockdefx[TimedAction]{Timed}{EndTimed}%
	[3]{\textbf{\size #1}%
	    \Prec#2\EndPrec%
	    \Eff#3\EndEff}{\bf Ending}
	    
\algnotext{EndPart} 
\algnotext{EndSmallPart}
\algnotext{EndPrec}
\algnotext{EndEff} 
\algnotext{EndIf} 
\algnotext{EndFor}
\algnotext{EndInput}
\algnotext{EndOutput}
\algnotext{EndInternal}
\algnotext{EndTimed}

\fi

\def\compactify{\itemsep=0pt \topsep=0pt \partopsep=0pt \parsep=0pt}
\let\latexusecounter=\usecounter

\newenvironment{CompactEnumerate}
  {\def\usecounter{\compactify\latexusecounter}
   \begin{enumerate}}
  {\end{enumerate}\let\usecounter=\latexusecounter}

\begin{document}

\conferenceinfo{MobiShare'06,} {September 25, 2006, Los Angeles, California, USA.}
\CopyrightYear{2006}
\crdata{1-59593-558-4/06/0009}

\pagestyle{plain}

\title{Energy Aware Self-organizing Density Management in Wireless 
sensor networks}

\numberofauthors{3}

\author{
\alignauthor Erwan Le Merrer\\
       \affaddr{France Telecom and }\\
       \affaddr{IRISA, France }\\
	   \email{erwan.lemerrer@orange-ft.com}
\alignauthor Vincent Gramoli, \\ Anne-Marie Kermarrec, \\ 
 Aline C. Viana \\
       \affaddr{IRISA/INRIA, France}\\
       \email{Aline.Viana@irisa.fr}
\alignauthor Marin Bertier \\
       \affaddr{IRISA/INSA, France}\\
       \email{Marin.Bertier@irisa.fr}
}

\date{25 September 2006}
\maketitle

\begin{abstract}

Energy consumption is the most important factor that determines sensor
node lifetime.  The optimization of wireless sensor network lifetime
targets not only the reduction of energy consumption of a single
sensor node but also the extension of the entire network lifetime. We
propose a simple and adaptive energy-conserving topology management
scheme, called SAND (Self-Organizing Active Node Density). SAND is
fully decentralized and relies on a distributed probing approach and
on the redundancy resolution of sensors for energy optimizations,
while preserving the data forwarding and sensing capabilities of the
network. We present the SAND's algorithm, its analysis of convergence,
and simulation results. Simulation results show that, though slightly
increasing path lengths from sensor to sink nodes, the proposed scheme
improves significantly the network lifetime for different neighborhood
densities degrees, while preserving both sensing and routing fidelity.

\end{abstract}




\section{Introduction}
\label{sec:introduction}

\noindent{\bf Context.}  Area monitoring is one of the most typical
applications of wireless sensor networks (WSNs). It consists in
deploying a large number of sensors in a given geographic area, for
collecting data or monitoring events.  It is not unusual that in this
situation, human intervention is not feasible. Sensors are then thrown
in mass, for example from a plane, and must be able to form a network
and to operate in a decentralized self-organized manner, maintaining
connectivity and area monitoring as long as possible. In addition,
because of the absence of wire and the small physical size of sensors,
WSNs have strong power restrictions.  Mechanisms for energy
optimization in WSNs constitute then an important requirement. This
optimization targets not only the reduction of energy consumption of a
single sensor node but also the extension of the entire network
lifetime.  The sensor network lifetime is defined as the period during
which the routing fidelity and the sensing fidelity of the network are
guaranteed. Guaranteeing sensing fidelity means that any monitored
stimulus in the area will always be sensed by at least one sensor.
Routing fidelity means the existence of a path between any
sensor node and {\it at least} one base station. Our
goal is then to leverage node redundancy in WSNs to reduce and distribute
the computational and communication energy consumption of the network
between sensors.  We consider that the cooperative nature of sensors
offers significant opportunities to manage energy consumption.

Given the potentially large number of sensors in a WSN and their
limited resources, it is also crucial to deploy fully decentralized
solutions and to evenly spread the load over the network. Recent works
in Peer-to-Peer (P2P) systems can be successfully adapted to explore
new solutions for energy consumption distribution in sensor networks.

\noindent{\bf Contributions.}  This paper describes an adaptive and
fully decentralized topology management scheme, called SAND
(Self-organizing Active Node Density). SAND significantly extends
network lifetime by reducing nodes activity. 
The major contribution of SAND is its simple algorithm: 
each node takes local decision based on the observation of
its neighborhood in order to ensure the properties of routing and sensing
fidelity. Moreover, SAND as a whole converges towards the
expected properties. In order to manage energy consumption, we
leverage P2P's cooperation paradigms and explore local node
information and sensors resolution through neighborhood communication
only.  No node location information is required, and our approach is
independent of any wireless routing protocol.

 
For any configuration of the
energy-aware topology, we show that SAND guarantees that 
routing and sensing fidelity will be extended. In comparison with the
case where no topology management is applied, SAND considerably
extends network lifetime at the price of the slightly increasing paths
lengths from sensor to sink nodes.  Moreover, simulations of SAND
suggest that network lifetime and SAND's robustness increases
proportionally to node density.

\noindent{\bf Outline.} The rest of this paper describes and evaluates 
SAND. In Section~\ref{sec:foreword}, we discuss 
the advantages in performing power management of radios in wireless 
networks and describe our system model. 
We review related works in Section~\ref{sec:related}. In 
Section~\ref{sec:design}, we describe the SAND approach, analyse interesting 
properties, and 
discuss design issues.  
We present the performance results in Section~\ref{sec:perform}. Finally, in
Section~\ref{sec:conclusion}, we conclude this paper and discuss future 
works.
\section{Foreword}
\label{sec:foreword}

\subsection{Where does the energy go?}
\label{subsec:power_go}

\begin{table*}[t]
\caption{Energy state description} 
\label{tab:Table1}
 \begin{center}
\scriptsize{
 \begin{tabular}{|c|c|c|l|} \hline
  {\bf State}   & {\bf Energy consumption}      & {\bf Activity}  & {\bf Provided service}\\ \hline

{\it Sleep}         & very low  & none    & - periodically turns on the radio for 
					receiving control msg \\ \hline

{\it Sensor-only}   & low       & sensing & - sensing \\ 
                    &           &         & - periodically turns on the radio for 
					receiving/sending control msg\\ 
                    &           &         & - sending local data, when needed \\  \hline
					
{\it Gateway}       & high      & sensing & - sensing    \\ 
                    &           & routing & - receiving/sending  
					control msg\\ 
                    &           &         & - sending/forwarding  
					data \\ \hline

{\it Router-sensor} & high      & sensing & - sensing    \\ 
                    &           & routing & - receiving/sending 
					control msg\\ 
                    &           &         & - sending/forwarding 
					data \\ 
                    &           &         & - managing sensor      
					density     \\ \hline

  \end{tabular}}
 \end{center}
\vspace{-0.7cm}
\end{table*}

Power dissipation analyses of a sensor node in the literature show that 
wireless communication is a major energy consumer during system 
operation~\cite{raghunathan02,estrin99}. 
Results show that ({\it i}) the overhearing process increases power
consumption,\footnote{\scriptsize {In~\cite{tsiatsis01}, authors
show that for typical sensor network scenarios, around 65\% of all
packets received by a sensor node need to be forwarded to other
destinations.}} and ({\it ii}) energy optimizations must turn off the
radio and not simply reduce packet transmission and reception.  SAND,
therefore, incorporates power management into the communication
process. We optimize energy consumption by completely turning off
radio whenever possible, conserving energy both in {\it Idle} state
when no traffic exists and in overhearing due to data transfer.

We also take advantage of the sensor network density for 
energy savings. 
SAND also powers down sensor nodes that
are equivalent from a sensing perspective. 
In summary, SAND is a very simple topology management scheme, which
generates low communication overhead. SAND applies in the context of nodes 
that are able to:

\begin{CompactEnumerate}
\item turn off their radios for {\it communication power conservation}, 
while still maintaining connectivity between sensors and sinks, {\it i.e.}
the routing fidelity;
\item completely power down for {\it sensing power conservation}, while 
still ensuring a correct stimuli monitoring, {\it i.e.} the sensing fidelity.
\end{CompactEnumerate}

\subsection{System model}
\label{subsec:system}

We consider a distributed system consisting of a finite set of $n$
sensors and $s$ sinks, each uniquely identified. Nodes (sensors or
sinks) are spread into a delimited area. We consider that both sensor
and sink nodes form a connected network or that all network partitions
contain at least one sink node.  Nodes may crash and recover during
the network lifetime, we do not consider Byzantine failure. We assume
that, although there is no need for synchronized clocks, there is an
upper bound on the drift rate of local clocks. We define $\delta$ as
an upper bound on the transmission time of a message between two
neighbors. We consider that sensing and communication ranges are equal.

SAND does not require any node location information.  We explore
density determination by assuming that nodes communicate only by
$1$-hop broadcast toward nodes in their neighborhood, corresponding to
their transmission range.  Finally, a node may be in one of the four
following energy states: sleep, sensor-only, router-sensor, and
gateway. Each state corresponds to a defined node activity and one or
several services provided to the network (see Table~\ref{tab:Table1}).
SAND guides nodes' energy state switching independently of the
underlying wireless routing protocol. Interactions between SAND and
the routing protocol are part of a future work.

\section{Related works}
\label{sec:related}

Topology management techniques, called SPAN~\cite{chen01} and
GAF~\cite{xu01}, have similar goals to those of SAND: they trade
network density for energy savings while preserving the forwarding
capacity of the network.  Nevertheless, they do not exploit the
absence of traffic in the active sensing state.  Besides the energy
consumption reduction in the sensing state, the
STEM~\cite{schurgers02} scheme also proposes to ensure a satisfactory
latency for transitioning to the routing state.  Authors then suggest
integrating STEM with scheme such as GAF and SPAN.  SAND coordinates
the radio sleep and the wakeup switch for sensing and/or routing
states in a unique scheme, being no integration
necessary.\footnote{\scriptsize {The insurance of a satisfactory
latency for transitions of state is not our focus here.}}
Several algorithms have been proposed for exploiting the area coverage
problem in sensor networks~\cite{carle05,tian02,simplot_percom06}.
Contrary to SAND, all these solutions assume that the sensors are
aware of their own positions. Authors in~\cite{tian02} do not address
the connectivity problem and require every sensor to know all their
neighbors positions before making its monitoring
decisions~\cite{tian02}. 
The proposal presented by 
Carle and Simplot-Ryl~\cite{carle05} specifies
that each sensor needs to
construct its subset of relays and broadcast it to its neighbors,
which generates higher communication overhead than SAND. The solution
proposed by A. Gallais {\it et al.}~\cite{simplot_percom06} relies on
low communication overhead and does not need any neighborhood
discovery. Nevertheless, nodes have to memorize the positions and the
decisions of their neighbors in order to make appropriate monitoring
decisions. In SAND, however, nodes need a small amount of information
({\it e.g.} only a partial neighborhood discovery) and perform a low
processing overhead to take their activity decisions.
Clustering algorithms can be also used to select router nodes. 
As an example, 
M. Chatterjee {\it et al.} in~\cite{chatterjee02} use 
quite sophisticated concepts and heuristics 
to decide which nodes should be cluster heads. SAND instead, has as a major 
contribution the simplicity for router election procedure.

\section{SAND design rationale}
\label{sec:design}

A major issue in sensor-based applications is the diffusion of the 
sensed data to a specific entity that can store and 
process it~-- the sink node. Mechanisms that ensure 
this diffusion are important components in WSNs and
have been the subject of many researches in the literature~\cite{alkaraki04}. 
In this context, the cooperation between energy conserving and 
information diffusion robustness is crucial. 

\begin{figure}[htbp]
  \centering
  \epsfig{file=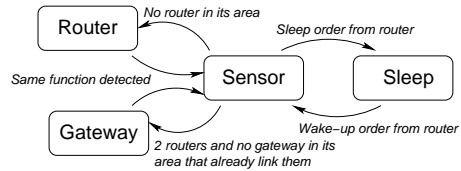, width=6cm}\vspace{-0.15cm}
  \caption{Energy state transition diagram in SAND.}
  \label{fig:state_diag}
\end{figure}

We take advantage of information provided by routing layer and/or
information related to the envisaged application to determine when the
radio is not needed.  We consider that during network lifetime, sensor
nodes can alternate their energy consumption between four states (see
Fig.~\ref{fig:state_diag}): (1) {\it sleep}, where all hardware
components are powered off, (2) {\it sensor-only}, where only sensor
and some pre-processing circuity are powered on, and (3) {\it gateway}
and {\it router-sensor}, where all hardware components are powered
on.\footnote{\scriptsize {Router-sensor and gateway nodes can also
optimize local energy consumption by changing the power state of their
memory and/or processor.}}  SAND performs then the energy-aware
topology management by controlling the routing and the sensing
fidelity during the network lifetime.

\subsection{Forwarding nodes distribution}
\label{subsec:routing}

The forwarding nodes distribution is performed in two consecutive
phases. The first one distributes nodes in router-sensor state
uniformly in the network. The second one consists in connecting close
router-sensors by selecting nodes to switch to the gateway
state. These two phases are based on the SONDe's
principle~\cite{lemerrer06}: if a node does not detect any neighbor in
each one of these two states then it turns itself into the missing
state.

\ifcode 
\begin{algorithm}[!ht]
  \caption{\textit{The SAND Algorithm}}\label{alg:sand}
  \begin{algorithmic}[1]
  {\scriptsize
	  \Part{Phase 1}{
				\If{$\neg\lit{router-detection}()$}																\label{line:nrdetect}
		    	\State $\ms{status} \gets \lit{router-sensor}$									\label{line:rsensor}
		    \ElsIf{$r \gets \lit{router-detection}()$}												\label{line:rdetect}
		    	\If{$r.\ms{ts} > \ms{ts}$}
		      	\State $\ms{status} \gets \lit{sensor-only}$									\label{line:back}
		      \EndIf
		  	\EndIf
		}\EndPart
		
		\Statex
		
		\Part{Phase 2}{																
		  	\If{$|\ms{routers} \gets \lit{detected-routers}()| \geq 2$}				\label{line:phase2-start}
		  	  \If{$\neg (g \gets \lit{gateway-detection}()) \vee (g.\ms{ts} < \ms{ts})$ \\
		  	  	\hspace*{4em} $\wedge \neg (\ms{routers} \subseteq \ms{g.routers})$}
			      	\State $\ms{status} \gets \lit{gateway}$	
			    \Else ~$\ms{status} \gets \lit{sensor-only}$										\label{line:phase2-sup}
			    \EndIf
				\ElsIf{$\ms{status} = \lit{gateway}$}										
					\State $\ms{status} \gets \lit{sensor-only}$
				\EndIf																														\label{line:phase2-end}
		}\EndPart
  }
  \end{algorithmic}
\end{algorithm}
\fi

\noindent {\bf Algorithm overview.}
The SAND algorithm is presented in Algorithm~\ref{alg:sand}.  
Each period of time $\Delta
> 2\delta$, the
router-sensor nodes send {\tt Hello} messages containing
their current state and a timestamp $\ms{ts}$. 
Observe that we
do not focus on specifying the message retransmission in case of
collision; we rather assume that this is implemented at a lower layer.
The timestamp $\ms{ts}$ of node $i$ contains: the time $i$ spent 
in its current state and its identifier (tie breaker).
Each sensor-only node in the
network checks if there is a router-sensor node in its immediate
neighborhood, by listening during a timeout $T_{on}$ with $T_{on} >
\Delta + \delta$.
If no router-sensor node is detected, then the sensor-only node becomes a
router-sensor node (Lines~\ref{line:nrdetect}-\ref{line:rsensor}).
If a router-sensor node detects the presence of another
router-sensor node in its transmission range with a 
higher timestamp than its own, it comes back to the 
sensor-only state (Lines~\ref{line:rdetect}-\ref{line:back}). 
While the first phase guarantees a good distribution of router-sensor, the 
second phase (Lines~\ref{line:phase2-start}-\ref{line:phase2-end}) elects
gateways to connect them. 
A sensor-only node becomes gateway if it
detects two routers and no gateway with a lower timestamp than its own
and that already makes a link path between those two routers. A
gateway node informs with a period $\Delta$ about the routers it
sees. If a gateway detects another gateway which connects the same
router-sensor and with a higher timestamp than its own, it comes back
to the sensor-only state (Line~\ref{line:phase2-sup}).

\noindent {\bf Independent-dominating set convergence.} In the
following, we show that SAND presents two interesting properties
borrowed from graph theory.  First, to help routing to a sink node,
each node is in the neighborhood of a router-sensor node or is itself
a router-sensor node.  Second, to prevent energy consumption waste, a
subset of nodes becomes router-sensor nodes.  There are solutions for
related problems known as \emph{vertex cover} and \emph{minimal
dominating set}, which guarantees activated sensors to form such a
set.

The minimality problem of the aforementioned solutions might involve
many state changes each time a router-sensor crashes.  Here, we rather
ensure that the sensor-router nodes satisfy both domination and
independence properties.  This helps at reducing, yet making it
sufficient, the number of router-sensor nodes, while this number has
not to be minimal.  Roughly speaking, the router-sensor nodes satisfy
\textit{(i)}~\emph{dominance}: all nodes are either router-sensor or a
neighbor of a router-sensor and \textit{(ii)}~\emph{independence}: no
router-sensor node is a neighbor of another router-sensor node.

The following proof shows that the algorithm converges to a
configuration verifying both properties under system stabilization.
Let a \emph{real-router}
(resp. \emph{real-sensor}) denote a router-sensor node
(resp. sensor-only node) that became router-sensor (resp. sensor-only)
$T_{on}$ time ago, and which did not revert its state since then.  
%

\begin{theorem}
	The SAND algorithm converges towards an independent-dominating set.
\end{theorem}
\begin{proof}[sketch]
 The proof is divided in three parts.  
 First we show that the independent-dominating property is an invariant. 
 Consider the communication subgraph containing only real-routers
 and real-sensors whose real-routers form an independent dominating
 set when the system stabilizes. 
 That is, any
 real-sensor node receives a message from a real-router in each
 $T_{on}$ period of time and does not become a router.  Similarly the
 real-routers stay in their state.

 Second we show that independence can never be violated.  Observe that
 nodes are initially in their sensor-only state and thus can not
 violate independence by awakening.  Before stabilization, real-router
 can crash but independence is never violated since message delay
 remains bounded.

 Finally, we show that the number of locations where the 
 dominance is violated eventually decreases, let ${\cal G}$ be any 
 communication subgraph whose set of real-routers is not dominant.  After some
 time, some nodes of ${\cal G}$ become routers.  This might happen in
 the meantime at different places in the same neighborhood.  After a
 $T_{on}$ delay, messages are exchanged between routers and
 the one with the lowest $\ms{ts}$ is chosen to become exclusively the
 real-router of the neighborhood.  Other router nodes, so as
 sensor-only nodes, become real-sensor. 
\end{proof}

Experiment described in Fig~\ref{fig:screenshots}(a) confirms our
theoretical analysis and shows the result of our router selection
after three  simulation rounds: this simulation leads to 50
router-sensor nodes (black circles) for 450 sensor-only nodes (gray
circles).
Provided these properties, we claim that a sufficient sensor nodes
density will provide enough gateway candidates to ensure the
connection between two close router-sensors. Consequently, SAND does
not determine the optimal minimum number of forwarding nodes to
maintain sink connectivities, ensuring then, that there are several paths between any
node and at least one sink. This redundancy makes the routing fidelity
more resilient to failures. Our claim is satisfied by the experiments
obtained in Fig.~\ref{fig:screenshots}(b) commented hereafter.
Specifying a protocol to obtain a path among all router-sensors from a
dominating set is left as an open work.

\subsection{Sensing guarantee}
\label{subsec:sensing}

Fidelity in stimuli sensing can be ensured only by router-sensor and
gateway nodes, because they are uniformly distributed in the network
and their sensing range is equal to their transmission
range. Nevertheless, for reliability issues and for the cases where a
specific sensor node density should be ensured, SAND allows the
control of the sensor-only nodes resolution in each target area, while
performing the sensing load distribution among nodes.

To this end, router-sensor nodes are in charge of selecting nodes to
switch to sleep or sensor-only state. This selection depends on the
envisaged reliability degree of the monitored area. Thus, nodes that
are for a long time in the sensor-only state will be selected to
switch to sleep state, and vice versa.  Each sensor-only node
periodically turns on its radio and sends {\tt Hello} messages
containing its current state and its estimated lifetime
($el$).\footnote{\scriptsize {$el$ corresponds to the expected
remaining node energy and is set by assuming that nodes will have a
start energy level and they will consume energy (related to their
state) until they die.}} Nodes in sleep state also periodically turn
on their radio, but never send messages.\footnote{\scriptsize {We
assume their estimated lifetime is the same of the last $el$ sent when
they were in the sensor-only state.}}

Upon reception of sensor-only $el$ of sensor-only neighbors,
router-sensor nodes compute the average (by considering the $el$ of
the last switched-to-sleep nodes too) and the standard deviation (by
considering only the sensor-only nodes that have lower $el$ than the
resulting average). Finally, router-sensor nodes send $1$-hop ``{\tt
switch-to-sleep}'' order messages to sensor-only nodes that have their
$el$ level lower than the resulting standard deviation.  The switch
of sensor-only to sleep state is performed as soon as another
sensor-only node appears in the monitored area.
Sensor-only nodes also turn on their radio if any local collected data
has to be transmitted to the sink ({\it e.g.}, for full memory
resource).

Router-sensor nodes also control the switch from 
sleep to sensor-only states, by sending ``{\tt switch-to-sensor}'' order 
messages.\footnote{\scriptsize {In the 
case of the reception of two contrary order 
messages, the priority is given to the ``{\tt switch-to-sensor}'' 
messages.}}
In this case, nodes in the sleep state switch to the sensor-only state
if they have (1) their radio turned on, and (2) their $el$ level is
higher than the computed average. In addition, like sensor-only nodes,
nodes in the sleep state with the radio turned on, verify their
router-sensor node connectivity.  If no {\tt Hello} message from a
router-sensor node is received, a sleep node becomes a sensor-only
node and then, based on the SAND bootstrap procedures, can switch to
the router-sensor or gateway state.  Fig.~\ref{fig:screenshots}(b)
shows the result of the SAND energy state distribution at a random
point in time.
The figure presents 68 router-sensor nodes (black circles), 47 gateway
nodes (squares) and 385 sensor-only and sleep nodes (gray circles).
The line connecting points represents the connectivity among
forwarding nodes.  The network connection among forwarding, sleep, and
sensor-only nodes is not represented here.
This shows 
that the generated forwarding topology is connected.

\begin{figure}[t]
   \centerline{
	 		\resizebox{1.45in}{!}{
         \includegraphics[scale=0.3,clip=true]{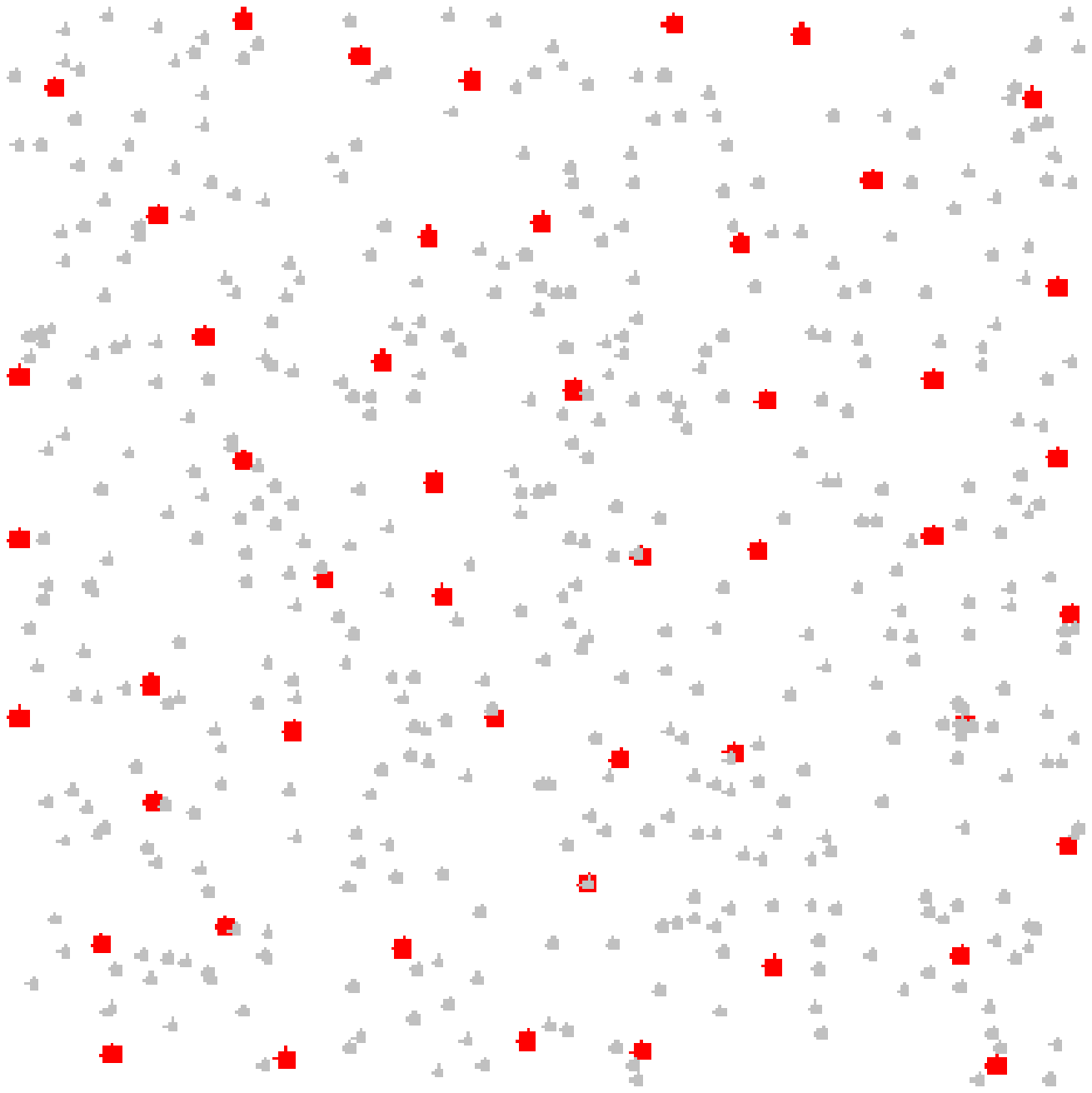}
       }\hspace{0.5cm}
       \resizebox{1.6in}{!}{
         \includegraphics[scale=0.3,clip=true]{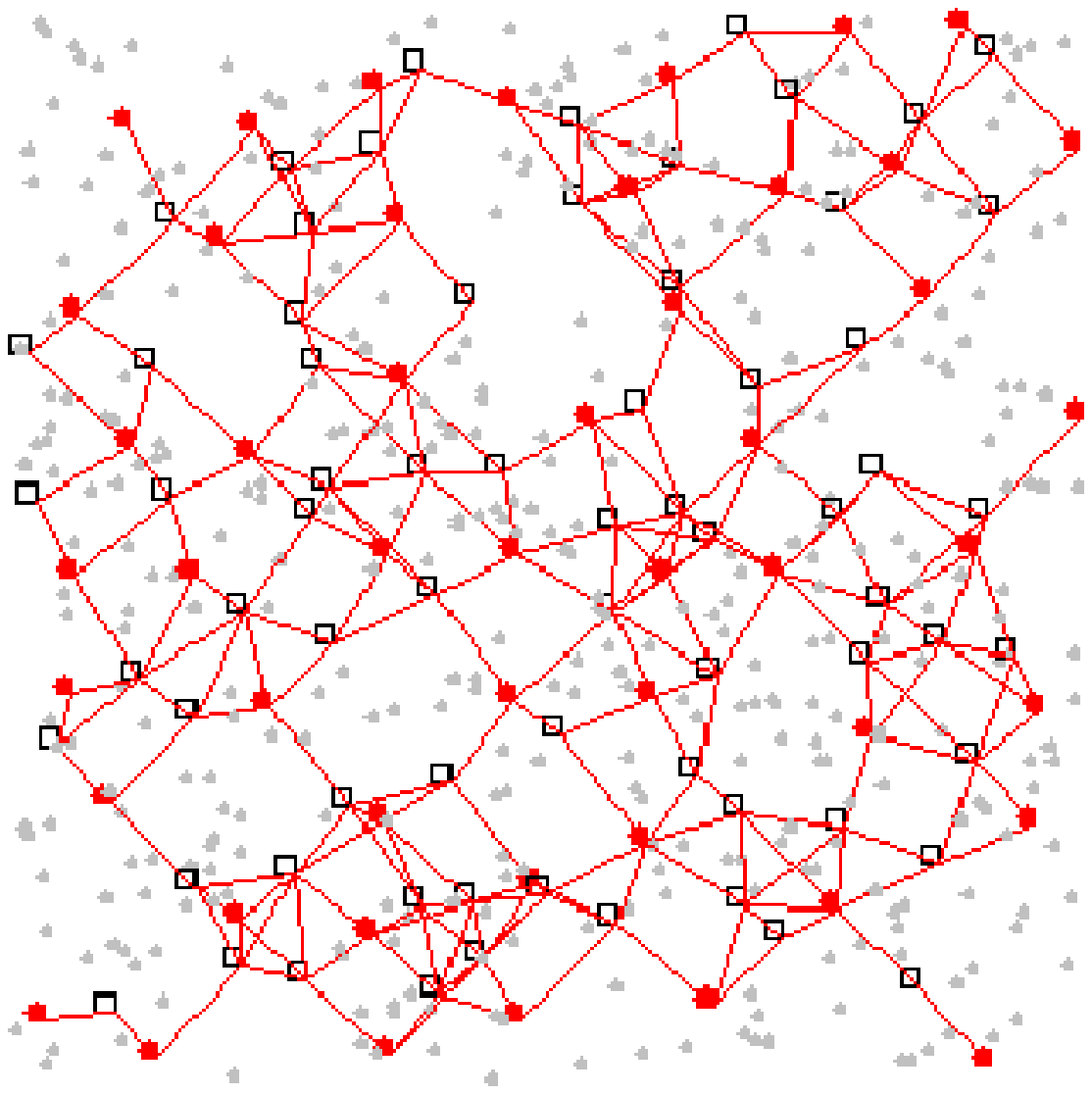}
       }}
   \vspace{-0.3cm}
\caption{A 300$\times$300 $m$ area network with an average  
   of 20 neighbors per node representing: (a)~the router selection, 
   (b)~the energy state distribution.}
   \label{fig:screenshots}
\vspace{-0.5cm}
\end{figure}


\subsection{Discussion}
\label{subsec:discussion}
We now discuss some design choices.

{\bf Outlining parameters:} The radio of sleep and sensor-only nodes
is periodically turned on for fidelity verifications at intervals
$T_{off}$ and stay on for at least a timeout value $T_{on}$.  The
range of $T_{off}$ can be influenced by the time that nodes have been
conserving their energy during the sleep state. Optimizations of
these parameters are under evaluation.

{\bf Outlining advantages:} Contrary to some existing area coverage
solutions, SAND does not require nodes position information or
geographic coverage computing for ensuring connectivity. Therefore SAND
has a low computing overhead.  Moreover, since nodes do not need to
perform a complete neighborhood discovery to take their state
decisions, SAND has also as advantage a low communication overhead.


{\bf What if disconnections occur:} Considering the poor failure
resilience of sensor nodes and the SAND's guarantee of  sink
connectivity, it may occur that failed nodes cause temporary
disconnections. This can also occurs if {\tt Hello} messages from
router-sensor nodes are lost, causing gateway nodes to become
sensor-only nodes.  Nevertheless, since SAND allows disconnection
detection and restores connectivity by state switching, we consider
that disconnection periods will not be long.  Nodes may store their
messages and wait for connectivity restoration. In the case of
router-sensor {\tt Hello} message losses, gateway nodes can be instrumented
to wait at least a timeout value of $2\times T_{on}$ before deciding
to switch to sensor-only state. This kind of improvements is subject
to future works.

\section{Performance evaluation}
\label{sec:perform}

We have conducted a number of simulation experiments over  
800 rounds, using a simulator consisting of the SAND
engine and a network emulation environment. We experiment
on a large scale static network with 4,700 nodes. 
Nodes are uniformly distributed over a square area of 700 $m$ on a 
side and have a transmission range of 37 $m$. 
The dynamism of the created topology is imposed by fail nodes only.

The simulator is a discrete time-based engine, in
which actions are performed per round of simulation. We have set
$T_{off}$ to be 2 rounds. We consider that nodes send order or {\tt
Hello} messages at each round. We also consider that sensor-only nodes
send {\tt Hello} messages at each $\Delta$ period, where $\Delta$ and
$T_{on}$ are initially set to 1 round.

We set all nodes with an 
estimated lifetime $el$ of 100,000 unities ($u$). 
Our energy consumption model 
is based on the power consumption of the sensor node described 
in~\cite{sinha01}. We consider thus, the costs described in 
Table~\ref{tab:Table2}. The network lifetime column shows
approximated values for the estimate maximum lifetime corresponding to
nodes in each state.

We uniformly generated 1,000 stimuli over the network.
Unless otherwise specified, we set the reliability
degree of each monitored stimulus to 5 nodes.  No wireless routing
protocol is implemented.  Thus, to evaluate the routing fidelity of
our approach, we simply verify if there is a path between a source
node ({\it i.e.}, a sensor-only, a gateway, or a router-sensor node)
and at least one sink node.  Path verification is performed each time
a stimulus is generated.  Unless otherwise specified, sink nodes are
uniformly distributed over the network and correspond to 1\% of 
all nodes.
To examine the impact of the neighborhood
density in the network lifetime, we vary the number of nodes from
1,000, to 1,900, 2,800, and 4,700 (corresponding to approximately 10,
20, 30, and 50 neighbors in range, respectively), while keeping
constant the area and the transmission range of nodes.

\vspace{-0.5cm}
\begin{table}[htbp]
\caption{Energy consumption} 
\label{tab:Table2}
\scriptsize{
 \begin{center}
 \begin{tabular}{|c|c|c|c|} \hline
  {\bf Node}       & {\bf Radio}   & {\bf Energy}      & {\bf Estimated} \\     
  {\bf state}      & {\bf state}   & {\bf consumption}   & {\bf lifetime}   \\ \hline 
{\it Sleep}         & radio OFF     & $\pm$ 10$u$              & $\cong$3300           \\ 
                    & radio ON      & $\pm$ 70$u$              & $rounds$           \\ \hline
                                                                           
{\it Sensor-only}   & radio OFF     & $\pm$ 200$u$             & $\cong$450         \\
                    & radio ON      & $\pm$ 270$u$            & $rounds$          \\ \hline
					                                                      
{\it Gateway} and      & radio ON    & $\pm$ 1040$u$           & $\cong$96              \\        
{\it Router-sensor}    &             &                  &  $rounds$               \\ \hline

  \end{tabular}
 \end{center}}
\end{table}
\vspace{-0.45cm}

We evaluate SAND along the following metrics:
\textit{(i)}~the network lifetime and the energy conservation; 
\textit{(ii)}~the forwarding robustness;
\textit{(iii)}~the sensing fidelity preservation; and
\textit{(iv)}~the effects of network density.

\noindent {\bf Experimental results.}
%
One of the SAND goals it
to preserve network routing fidelity.  We consider that if paths in
the gateway/router-sensor nodes backbone exist, there are similar
non-conflicting paths in the underlying network.
Fig.~\ref{fig:suc_pck_sink} evaluates the robustness of SAND in
ensuring a sink connectivity in a 1,900-nodes network. For this
purpose, we vary the sink node density from 0.5\%, 1\%, and 1.5\% of
the total number of nodes.  As expected, as sink resolution decreases,
more active nodes are needed to ensure sink connectivities. This
results in a decrease of the network lifetime and the faster energy
exhaustion of nodes.  Nevertheless, compared to results obtained in
Without-SAND topology with 1\% of sink density, SAND still extends the
number of forwarding paths of network (95\% of paths are ensured for a
double of time than in Without-SAND), even with a lower sink
resolution of 0.5\%.

Additionally, for different neighborhood densities, 
we show that despite using fewer forwarding nodes, SAND 
does not significantly increase the number of hops of paths to sink nodes.
Table~\ref{tab:Table3} shows the average ($\tau$) and the standard 
deviation ($\sigma$) path length of SAND 
and Without-SAND, calculated while 90\% of the generated stimuli are sensed and 
corrected forwarded. SAND constructs forwarding paths with 
only a slightly higher number of hops, on average 18\% of more 
hops. 

\vspace{-0.5cm}
\begin{table}[htbp]
\caption{Path length results.} 
\label{tab:Table3}
 \begin{center}
\scriptsize{
 \begin{tabular}{|c|c|c|c|c|c|} \cline{3-6}
\multicolumn{2}{c|}{} &\multicolumn{2}{|c|}{{\bf SAND}} &\multicolumn{2}{|c|}{{\bf W/o-SAND}} \\ \hline
{\bf Number} & {\bf Neighborhood} &         &          &          &   \\ 
{\bf of nodes } & {\bf density}   & $\tau$  & $\sigma$  & $\tau$  & $\sigma$ \\ \hline
1,000    &   10          & 3.146   & 1.844   & 2.566   & 1.843      \\ \hline
1,900    &   20          & 3.408   & 2.267   & 2.762   & 2.047           \\ \hline
2,800    &   30         &  3.043   & 2.086   & 2.461   & 2.057     \\ \hline
4,700    &   50         &  2.486   & 2.046   & 1.961   & 2.114          \\ \hline
  \end{tabular}}
 \end{center}
\end{table}
\vspace{-0.45cm}

Fig.~\ref{fig:energy_remain} shows the average of energy remaining at
each node after 40 simulation rounds under different neighborhood
densities.  In this simulation round, all nodes in the network are
still alive. The figure also compares the energy conservation resulted
in SAND with the case where no power management is performed. We
observe that SAND provides a considerable amount of energy saving over
Without-SAND. This is because all nodes in Without-SAND topology have
the same energy consumption as router-sensor nodes. In SAND, however,
nodes switch when ever possible to sleep or sensor-only states, as a
few forwarding nodes are present in each transmission range. We also
check that the energy saving increases proportionally to the
neighborhood density. This is due to the fact that as the density
increases, a lower fraction of alive nodes is composed of forwarding
nodes (the highest consumer of energy state), while a higher fraction
are in the sleep state (see Fig.~\ref{fig:active_alive_nodes}).

\begin{figure*}[!ht]
  \begin{center}
    \subfigure[]
    {
      \label{fig:suc_pck_sink}
      \epsfxsize=5.5cm
	  \leavevmode\epsfbox{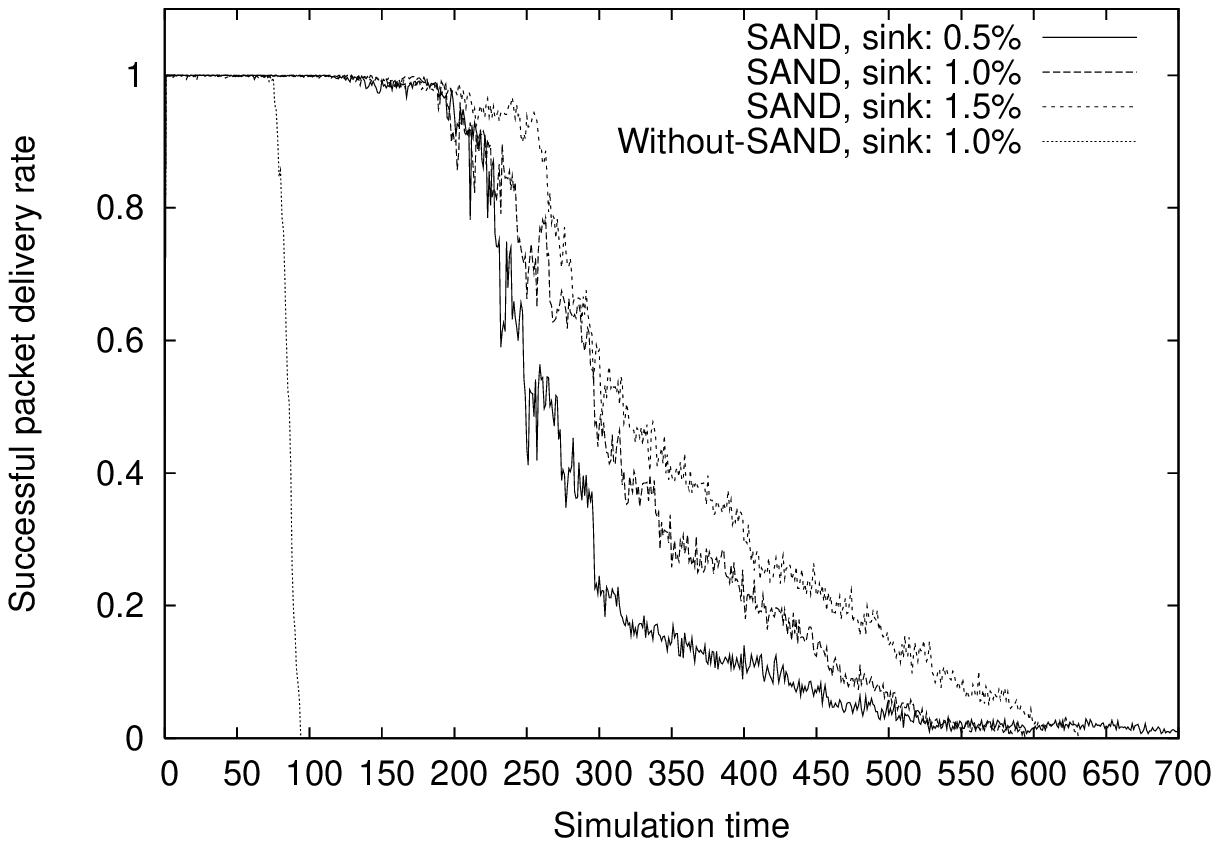}
    }
    \subfigure[]
    {
      \label{fig:energy_remain}
      \epsfxsize=5.5cm 
      \leavevmode\epsfbox{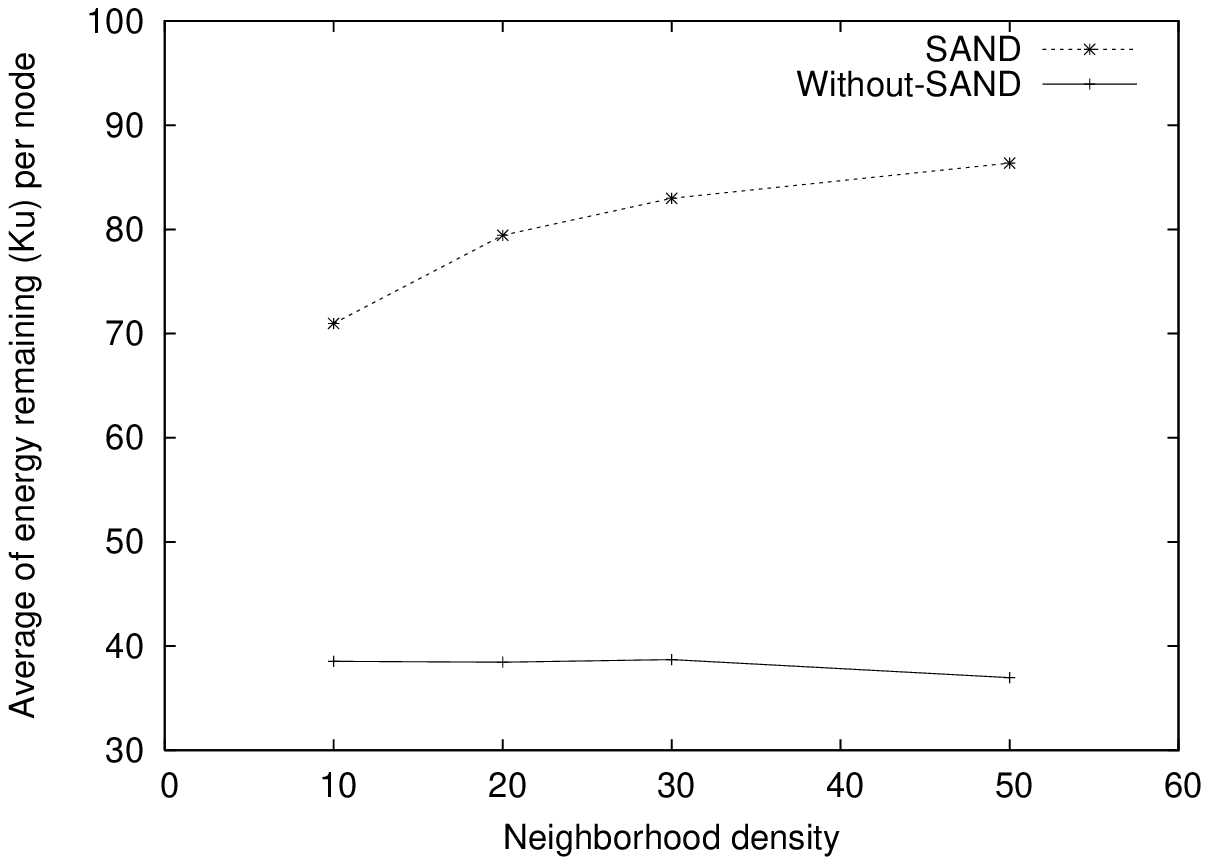}
    }
    \subfigure[]
    {
      \label{fig:succs_stimuli}
      \epsfxsize=5.5cm
      \leavevmode\epsfbox{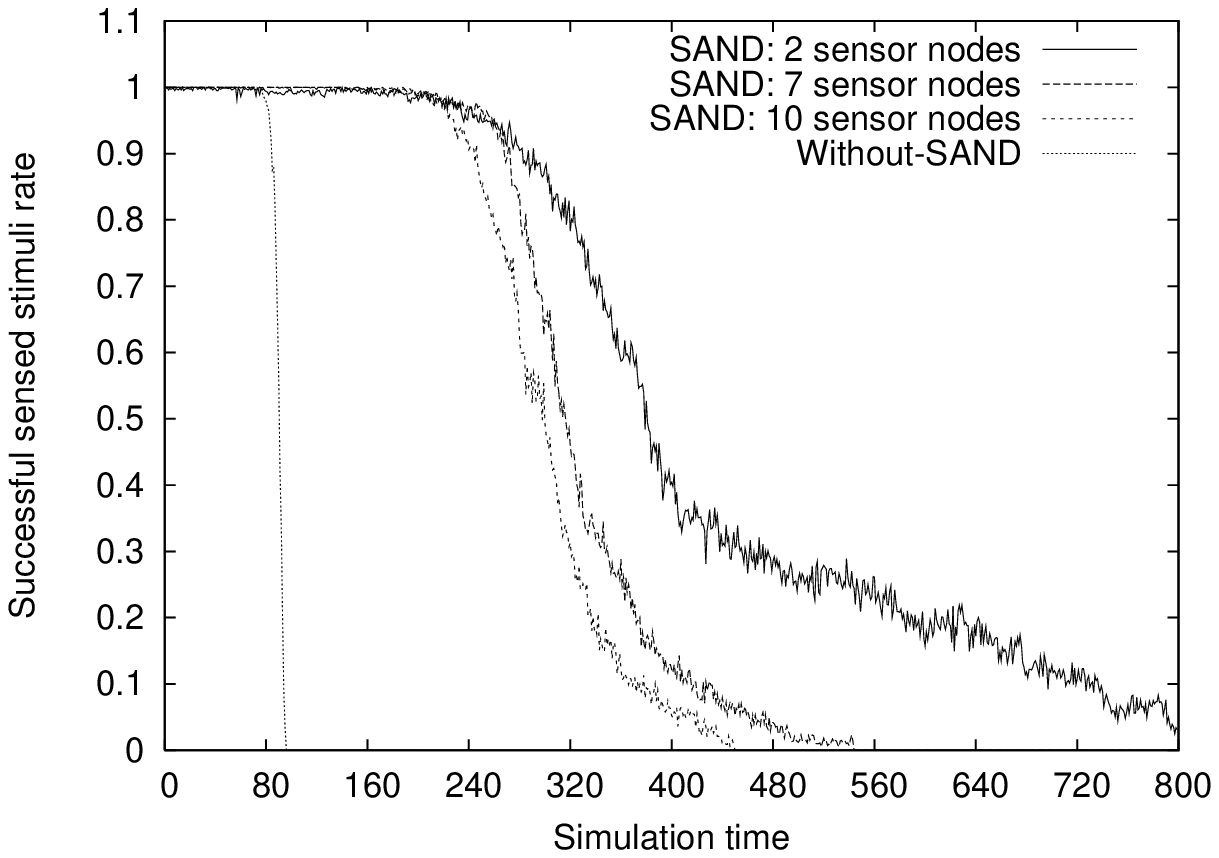}
    }\vspace{-0.25cm}
    \subfigure[]
    {
      \label{fig:pck_density_multisink}
      \epsfxsize=5.5cm
      \leavevmode\epsfbox{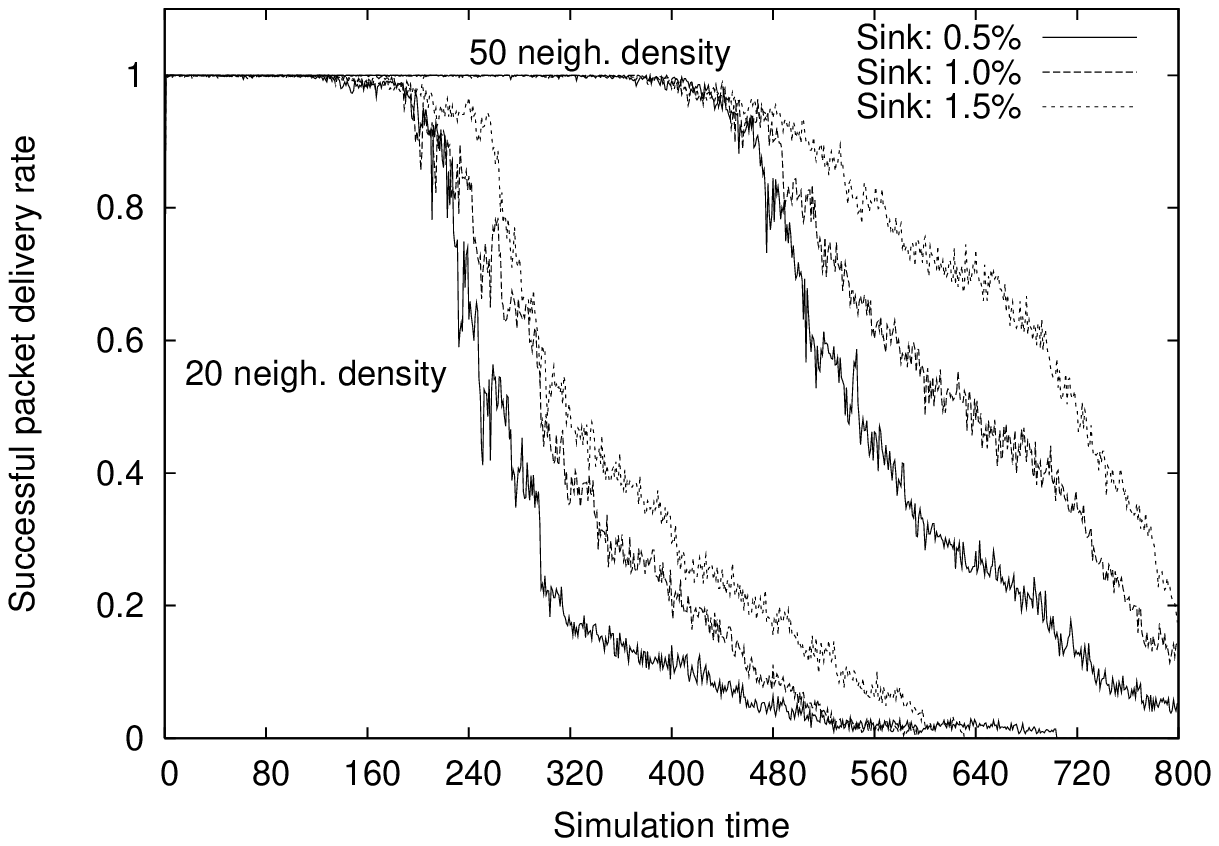}
    }
    \subfigure[]
    {
      \label{fig:active_alive_nodes}
      \epsfxsize=5.5cm
      \leavevmode\epsfbox{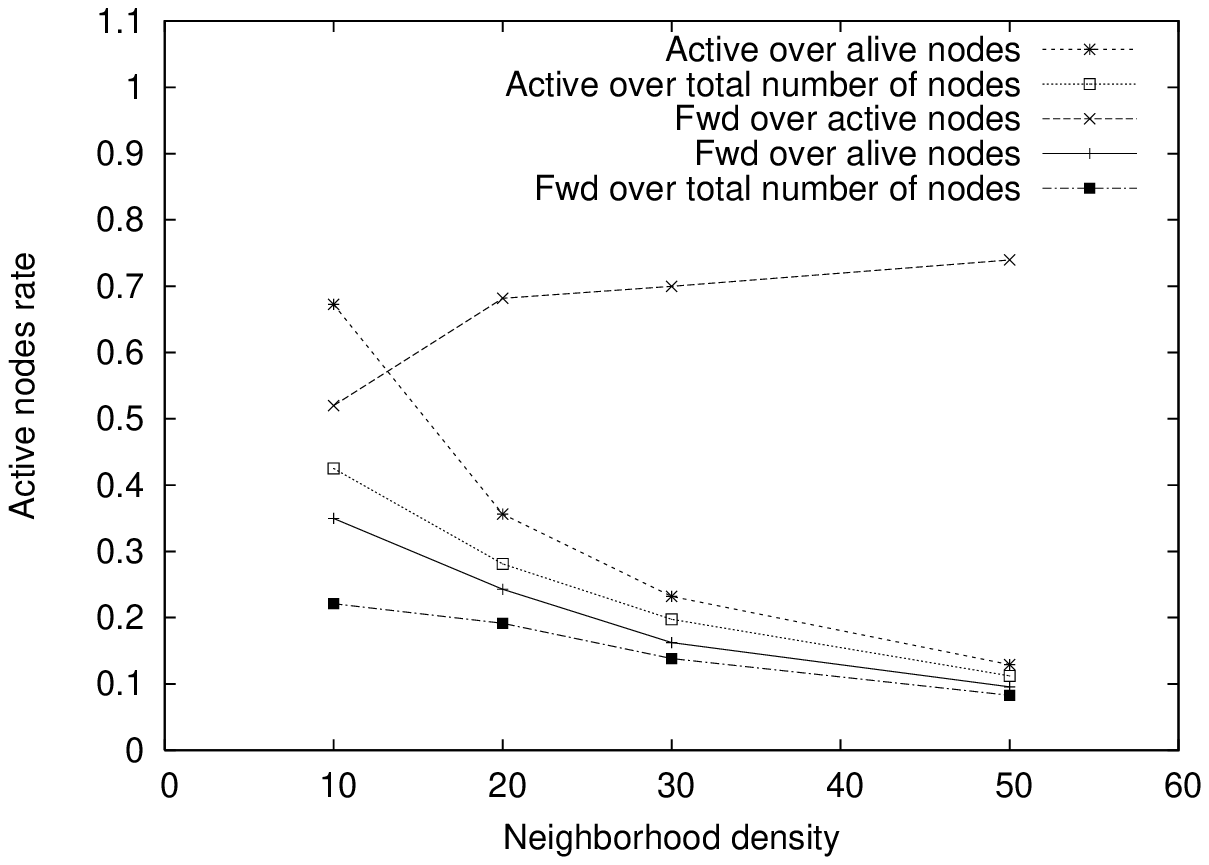}
    } \vspace{-0.4cm}
   \caption{(a) The sink connectivity. 
   (b) Average of energy remaining after 40 rounds. 
   (c) The stimuli sensing fidelity.
   (d) Successful forwarding paths.
   (e) Fraction of active and forwarding nodes.}
    \label{fig:graph}
  \end{center}
\vspace{-0.7cm}
\end{figure*}

Fig.~\ref{fig:succs_stimuli}
evaluates the robustness of SAND with respect to sensing fidelity, in
a 1,900-nodes network. We vary the required sensor node density in the
network from 2, 7, and 10 nodes, and show the rate of success sensed
stimuli for each density. Compared to the results obtained in
Without-SAND topology (where all nodes are sensing for all the
simulation time), SAND extends the number of sensed stimuli over the
network lifetime, even when a lower sensor node resolution of 2 is
used per target area. SAND ensures for a triple of time, 95\% of
sensing fidelity than Without-SAND. As expected, the increase of the
number of required sensor density, impacts the energy consumption of
nodes and consequently, decreases the network lifetime.

Fig~\ref{fig:pck_density_multisink} shows the results of the
performance of SAND in ensuring sink connectivity for different
neighborhood densities (for 20 and 50 neighbors in range) and sink
resolution (for 0.5\%, 1\%, and 1.5\% of the total number of nodes).
As expected, SAND is less impacted by sink resolution when the
neighborhood density increases. This is because the increase of
neighborhood density, increases the number of forwarding candidates,
and consequently, increases the probability of finding paths to sink
nodes.

Fig~\ref{fig:active_alive_nodes} shows the fraction of active and
forwarding nodes after 120 rounds of simulation (among many
experiments), under different neighborhood densities. Here, we compare
the fraction of active and forwarding nodes over the total number of
nodes in the network and over the total number of alive nodes. We show
that the potential of saving energy of SAND depends on the node
density, since the fraction of active and forwarding nodes depend on
the number of nodes per radio coverage area. We can easily verify that
the higher the node density, the lower the percentage of active nodes
in the network. Also the rate of forwarding nodes over the number of
active nodes decreases as the neighborhood density increase. On the
other hand, despite the increase of the rate of forwarding nodes over
the total number of nodes, the number of forwarding nodes becomes
almost constant for higher neighborhood densities ({\it e.g.}, for 30
and 50 neighbors in range).

\section{Conclusion}
\label{sec:conclusion}

In this paper, we have presented SAND, an approach to energy
conservation for wireless sensor networks. Energy consumption is one
of the most important factors that determines sensor node lifetime.
SAND is a fully decentralized, simple, and efficient algorithm able to
significantly extend not only sensor lifetime but also the entire
network lifetime. SAND focuses on turning off the nodes radio as much
as possible while still ensuring stimuli sensing and multi-hop routing
fidelity. We have presented the algorithm analysis and the simulation
of SAND. Our experiments show that SAND guarantees for a longer
time, (1) the existence of paths between any sensor node to at
least one sink node in the network and (2) the correct sensing of
stimulus in a monitored sensor network. SAND improves considerably
network lifetime proportionally to node density, at the price of the
slightly increasing paths length from sensor to sink nodes. Additional
analyses are being performed to evaluate the network behaviour of SAND
as nodes move. We also intend to study network partitions and to
validate our results with physical hardware in real scenarios.

\end{document}